\RequirePackage{fix-cm}
\documentclass[smallextended]{svjour3}       % onecolumn (second format)
\smartqed  % flush right qed marks, e.g. at end of proof
\usepackage{graphicx}
\usepackage{tikz}
\usetikzlibrary{calc,shapes.multipart,chains,arrows}
\usepackage{graphicx}
\usepackage{caption}
\usepackage{enumitem}
\usetikzlibrary{shapes.symbols,trees,positioning,fit,arrows,decorations.pathreplacing, backgrounds}
\usepackage{geometry}
\usepackage{url}
\usepackage{comment}
\usepackage{amsmath}
\usepackage{amsfonts}
\usepackage{amssymb}
\newcounter{linecounter}
\newcommand*{\procline}[1]{{\bf \refstepcounter{linecounter}\thelinecounter\label{ln:#1}.}}
\newcommand*{\refln}[1]{{\bf \ref{ln:#1}}}
\newcommand{\true}{true}
\newcommand{\false}{false}

\newcommand{\ifcode}{{\bf if} \,}

\newcommand{\elsecode}{{\bf else} \,}
\newcommand{\goto}{{\bf go to} \,}
\newcommand{\waittill}{{\bf wait till} \,}
\newcommand{\return}{{\bf return} \,}

\newcommand{\cas}{\mbox{CAS}}

\newcommand{\csowner}{\mbox{\sc CSStatus}}
\newcommand{\seq}{\mbox{\sc Seq}}
\newcommand{\token}{\mbox{\sc Token}}

\newcommand{\minarray}{\mbox{\sc Registry}}
\newcommand{\registry}{\minarray}
\newcommand{\go}{\mbox{\sc Go}}
\newcommand{\abort}{\mbox{\sc AbortSignal}}

\newcommand{\pc}[1]{\ensuremath{PC_{#1}}}
\newcommand{\pcp}{\pc{p}}
\newcommand{\peer}[1]{\ensuremath{peer_{#1}}}
\newcommand{\peerp}{\peer{p}}

\newcommand{\tok}[1]{\ensuremath{tok_{#1}}}
\newcommand{\tokp}{\tok{p}}

\newcommand{\bit}[1]{\ensuremath{b_{#1}}}
\newcommand{\bitp}{\ensuremath{\bit{p}}}
\newcommand{\myseq}[1]{\ensuremath{s_{#1}}}
\newcommand{\myseqp}{\ensuremath{\myseq{p}}}
\newcommand{\mygo}[1]{\ensuremath{g_{#1}}}
\newcommand{\mygop}{\ensuremath{\mygo{p}}}

\newcommand{\mwrite}{\ensuremath{{\tt write}}}
\newcommand{\findmin}{\ensuremath{{\tt findmin}}}
\newcommand{\promote}[1]{\ensuremath{{\tt promote}_{#1}}}
\newcommand{\promotep}{\ensuremath{\promote{p}}}
\newcommand{\tryproc}[1]{\ensuremath{\texttt{try}_{#1}}}
\newcommand{\tryprocp}{\ensuremath{\tryproc{p}}}
\newcommand{\exitproc}[1]{\ensuremath{\texttt{exit}_{#1}}}
\newcommand{\exitprocp}{\ensuremath{\exitproc{p}}}
\newcommand{\recoverproc}[1]{\ensuremath{\texttt{recover}_{#1}}}
\newcommand{\recoverprocp}{\ensuremath{\recoverproc{p}}}
\newcommand{\abortproc}[1]{\ensuremath{\texttt{abort}_{#1}}}
\newcommand{\abortprocp}{\ensuremath{\abortproc{p}}}

\newcommand{\limplies}{\Rightarrow}

\def\final{True}
\newcommand{\cond}[1]{%
	\label{inv:abrt:cond#1}
	\ifnum\pdfstrcmp{\final}{True}=0 \unskip\else #1 \fi\ignorespaces}

\newcommand{\resolve}[2]{\mbox{\refln{#1}::\refln{#2}}}

\newcommand{\gotocs}{\ensuremath{\mbox{IN\_CS}}}
\newcommand{\gotorem}{\ensuremath{\mbox{IN\_REM}}}
\newcommand{\isaborting}[1]{\ensuremath{flag_{#1}}}
\newcommand{\isabortingp}{\ensuremath{\isaborting{p}}}
\newcommand{\status}{\mbox{\em status}}
\newcommand{\good}{\mbox{\em good}}
\newcommand{\abortsigp}{\ensuremath{\abortsig[p]}}
\newcommand{\rectry}{\mbox{\em recover-from-try}}
\newcommand{\reccs}{\mbox{\em recover-from-cs}}
\newcommand{\recexit}{\mbox{\em recover-from-exit}}
\newcommand{\recrem}{\mbox{\em recover-from-rem}}
\newcommand{\procset}{\ensuremath{{\cal P}}}
\newcommand{\varset}{\ensuremath{{\cal X}}}
\newcommand{\valset}{\ensuremath{\mbox{\em Vals}}}
\newcommand{\initmap}{\ensuremath{{\cal F}}}
\newcommand{\operations}{\ensuremath{\mbox{\em OP}}}
\newcommand{\methods}{\ensuremath{{\cal M}}}
\newcommand{\partition}{\Delta}
\newcommand{\abortsig}{\mbox{\sc AbortSignal}}

\newcommand{\normal}{\mbox{\em normal}}

\newcommand{\crash}{\mbox{\em crash}}
\newcommand{\remproc}[1]{\ensuremath{\texttt{remainder}_{#1}}}
\newcommand{\remprocp}{\ensuremath{\remproc{p}}}
\newcommand{\critsec}[1]{\ensuremath{\texttt{cs}_{#1}}}
\newcommand{\critsecp}{\ensuremath{\critsec{p}}}
\newcommand{\upto}{\texttt{-}}

\newcommand{\cache}{\mbox{\sc Cache}}

\journalname{Special issue on NETYS 2019}
\begin{document}

\title{Recoverable Mutual Exclusion with Abortability
	\thanks{The first author is grateful to the Frank family and Dartmouth College
	for their support through James Frank Family Professorship of Computer Science. 
	The second author is grateful for the support from Dartmouth College.}
}


\author{Prasad Jayanti         \and
        Anup Joshi %etc.
}


\institute{Prasad Jayanti \at
              Dartmouth College,
              Hanover NH 03755, 
              USA\\
              \email{prasad.jayanti@dartmouth.edu}           %  \\
           \and
           Anup Joshi \at
              Dartmouth College,
              Hanover NH 03755, 
              USA\\
              \email{anup.s.joshi.gr@dartmouth.edu}
}

\date{Received: date / Accepted: date}

\maketitle

\begin{abstract}
Recent advances in non-volatile main memory (NVRAM) technology 
have spurred research on designing algorithms that are resilient to process crashes.
This paper is a fuller version of our conference paper \cite{jayanti:rmeabort},
which presents the first Recoverable Mutual Exclusion (RME) algorithm that supports abortability.
Our algorithm uses only the read, write, and CAS operations, which are commonly supported by multiprocessors.
It satisfies FCFS and other standard properties.

Our algorithm is also adaptive.
On DSM and Relaxed-CC multiprocessors, a process incurs $O(\min(k, \log n))$ RMRs in a passage and
$O(f+ \min(k, \log n))$ RMRs in an attempt, where 
$n$ is the number of processes that the algorithm is designed for,
$k$ is the point contention of the passage or the attempt,
and $f$ is the number of times that $p$ crashes during the attempt.
On a Strict CC multiprocessor, the passage and attempt complexities are $O(n)$ and $O(f+n)$.

Attiya et al. proved that, with any mutual exclusion algorithm,
a process incurs at least $\Omega(\log n)$ RMRs in a passage,
if the algorithm uses only the read, write, and CAS operations \cite{Attiya:lbound}.
This lower bound implies that the worst-case RMR complexity of our algorithm is optimal
for the DSM and Relaxed CC multiprocessors.

\end{abstract}

\section{Introduction}

Recent advances in non-volatile main memory (NVRAM) technology \cite{inteloptane}\cite{nvmm:pcm}\cite{nvmm:memristor}\cite{nvmm:mram} 
have spurred research on designing algorithms that are resilient to process crashes.
NVRAM is byte-addressable, so it replaces main memory, directly interfacing with the processor.
This development is exciting because, if a process crashes and subsequently restarts, 
there is now hope that the process can somehow recover from the crash by consulting the contents of the NVRAM
and resume its computation.

To leverage this advantage given by the NVRAM, there has been keen interest in reexamining the important distributed computing problems
for which algorithms were designed in the past for the traditional (crash-free) model of an asynchronous shared memory multiprocessor.
The goal is to design new algorithms that guarantee good properties even if processes crash at arbitrary points 
in the execution of the algorithm and subsequently restart and attempt to resume the execution of the algorithm.
The challenge in designing such ``recoverable'' algorithms stems from the fact that
when a process crashes, even though the shared variables that are stored in the NVRAM are unaffected, 
the crash wipes out the contents of the process' cache and CPU registers, including its program counter.
So, when the process subsequently restarts, it can't have a precise knowledge of exactly where it crashed.
For instance, if the last instruction that a process executes before a crash is a compare\&swap (CAS) on a shared variable $X$,
when it subsequently restarts, it can't tell whether the crash occurred just before or just after executing the CAS instruction and,
if it did crash after the CAS, it won't know the response of the CAS (because the crash wipes out the register the CAS's response went into).
The ``recover'' method, which a process is expected to execute when it restarts,
has the arduous task of ensuring that the process can still somehow resume the execution of the algorithm seamlessly.

The mutual exclusion problem, formulated
to enable multiple processes to share a resource that supports only one process at a time \cite{Dijkstra:mutex},
has been thoroughly studied for over half a century for the traditional (crash-free) model,
but its exploration for the crash-restart model is fairly recent.
In the traditional version of the problem, each process $p$ is initially in the ``remainder'' section.
When $p$ becomes interested in acquiring the resource, it executes the $\tryprocp()$ method;
and when this method completes, $p$ is in the ``critical section'' (CS).
To give up the CS, $p$ invokes the $\exitprocp()$ method; and when this method completes, 
$p$ is back in the remainder section.
An algorithm to this problem specifies the code for the try and exit methods so that
at most one process is in the CS at any time and other desirable properties
(such as starvation freedom, bounded exit, and First-Come-First-served, or FCFS)
are also satisfied.
Golab and Ramaraju were the first to reformulate this problem
for the crash-restart model as {\em Recoverable Mutual Exclusion} (RME).
In the RME problem, a process $p$ can crash at any time and subsequently restart \cite{Golab:rmutex}.
If $p$ crashes while in try, CS, or exit, $p$'s cache and registers (aka local variables)
are wiped out and $p$ returns to the remainder section (i.e., crash resets $p$'s program counter to
its remainder section).
When $p$ restarts after a crash, it is required to invoke a new method, named $\recoverprocp()$,
whose job is to ``repair'' the adverse effects of the crash and send $p$ to where it belongs.
In particular, if $p$ crashed while in the CS,
$\recoverprocp()$ puts $p$ back in the CS (by returning $\gotocs$).
On the other hand, if $p$ crashed while executing $\tryprocp()$,
$\recoverprocp()$ has a choice---it can either roll $p$ back to the Remainder
(by returning $\gotorem$) or put it in the CS (by returning $\gotocs$).
Similarly, if $p$ crashed while executing $\exitprocp()$,
$\recoverprocp()$ has a choice of returning either $\gotorem$ or $\gotocs$.

Golab and Ramaraju made a crucial observation that if $p$ crashes while in the CS,
then no other process should be allowed into the CS until $p$ restarts and reenters the CS.
This {\em Critical Section Reentry} (CSR) requirement was strengthed by Jayanti and Joshi's
{\em Bounded CSR} requirement: if $p$ crashes while in the CS,
when $p$ subsequently restarts and executes the recover method,
the recover method should put $p$ back into the CS in a bounded number of its own steps \cite{jayanti:fcfsmutex}.
There has been a flurry of research on RME algorithms in the recent years 
\cite{chan:amortizedrme}\cite{dhoked:adaptiverme}\cite{Golab:rmutex2}\cite{Golab:rmutex3}\cite{Golab:rmutex}\cite{jayanti:fasasmutex}\cite{jayanti:rmesublog}\cite{jayanti:fcfsmutex}\cite{jayanti:rmeabort}\cite{morrison:abrtrme}.

Orthogonal to this development of recoverable algorithms,
motivated by the needs of real time systems and database systems,
Scott and Scherer advocated the need for mutual exclusion algorithms 
to support the ``abort'' feature, whereby a process in the try section can quickly quit the algorithm, 
if it so desires \cite{scott:abort}.
More specifically, if $p$ receives an abort signal from the environment while executing the try method,
the try method should complete in a bounded number of $p$'s steps and
either launch $p$ into the CS or send $p$ back to the remainder section.
In the past two decades, there has been a lot of research on abortable mutual exclusion algorithms 
for the traditional (crash-free) model \cite{}.

The possibility of crashes, together with the CSR requirement, renders abortability even more important 
in the crash-restart model,
yet there have been no {\em abortable} recoverable algorithms until the conference publication
of the algorithm in this submission \cite{jayanti:rmeabort}.
There has since been one more algorithm, by Katzan and Morrison \cite{morrison:abrtrme},
and we will soon compare the two algorithms.

\subsection{RMR complexity.}

{\em Remote Memory Reference} (RMR) complexity is the standard
complexity metric used for comparing mutual exclusion algorithms,
so we explain it here.
This metric is explained for the two prevalent models of multiprocessors---{\em Distributed Shared Memory} (DSM) and {\em Cache-Coherent} (CC) multiprocessors---as follows.
In DSM, shared memory is partitioned into $n$ portions, one per process, and
each shared variable resides in exactly one of the $n$ partitions.
A step in which a process $p$ executes an instruction on a shared variable $X$
is considered an RMR if and only if $X$ is not in $p$'s partition of the shared memory.

In CC, the shared memory is remote to all processes, but every process has a local cache.
A step in which a process $p$ executes an instruction $op$ on a shared variable $X$
is considered an RMR 
if and only if $op$ is {\em read} and $X$ is not in $p$'s cache, or
$op$ is any non-read operation (such as a {\em write} or {\em CAS}).
If $p$ reads $X$ when $X$ is not present in $p$'s cache,
$X$ is brought into $p$'s cache.
If a process $q$ performs a non-read operation $op$ while $X$ is in $p$'s cache,
$X$'s copy in $p$'s cache is deleted in the {\em Strict CC model},
but in the {\em Relaxed CC model} it is deleted only if $op$ changes $X$'s value.
Thus, if $X$ is in $p$'s cache and $q$ performs an unsuccessful CAS on $X$,
then $X$ continues to remain in $p$'s cache in the relaxed CC model.

A {\em passage} of a process $p$ starts when $p$ leaves the remainder section
and completes at the earliest subsequent time when $p$ returns to the remainder
(note that $p$ returns to the remainder either because of a crash or because of 
a normal return from try, exit or recover methods).
An {\em attempt} of $p$ starts when $p$ leaves the remainder
and completes at the earliest subsequent time when $p$ returns to the remainder ``normally,''
i.e., not because of a crash.
Note that each attempt includes one or more passages. 

The {\em RMR complexity of a passage} (respectively, {\em attempt}) of a process $p$
is the number of RMRs that $p$ incurs in that passage (respectively, attempt).

\subsection{Adaptive complexity.}

A process is {\em active} if it is in the CS, or executing the try, exit, or recover methods,
or crashed while in try, CS, exit, or recover and has not subsequently invoked the recover method.
The {\em point contention} at any time $t$ is the number of active processes at $t$.
The point contention of a passage (respectively, attempt) is the maximum point contention
at any time in that passage (respectively, attempt).
An algorithm is {\em adaptive} if the RMR complexity $r$ of each passage (or attempt) 
of a process $p$ is a function of that passage's (or attempt's) point contention $k$
such that $r= O(1)$ if $k = O(1)$.

\subsection{Our contribution.}

We present the first abortable RME algorithm.
Our algorithm is based on the ideas underlying two earlier CAS-based algorithms---one that is recoverable 
but not abortable \cite{jayanti:fcfsmutex} and another that is abortable but not recoverable \cite{jayanti:abrt}.
Our algorithm uses only the read, write, and CAS operations, which are commonly supported by multiprocessors.
It satisfies FCFS and other standard properties 
(starvation-freedom, bounded exit, bounded CSR, and bounded abort).
The algorithm's space complexity---the number of words of memory used---is $O(n)$.

Our algorithm is also adaptive.
On DSM and Relaxed CC multiprocessors, a process $p$ incurs $O(\min(k, \log n))$ RMRs in a passage and
$O(f+ \min(k, \log n))$ RMRs in an attempt, where 
$n$ is the number of processes that the algorithm is designed for,
$k$ is the point contention of the passage or the attempt,
and $f$ is the number of times that $p$ crashes during the attempt.
On a Strict CC multiprocessor, the passage and attempt complexities are $O(n)$ and $O(f+n)$.

Attiya et al. proved that, with any mutual exclusion algorithm
(even if the algorithm does not have to satisfy recoverability or abortability),
a process incurs at least $\Omega(\log n)$ RMRs in a passage,
if the algorithm uses only the read, write, and CAS operations \cite{Attiya:lbound}.
This lower bound implies that the worst-case RMR complexity of our algorithm is optimal
for the DSM and Relaxed CC multiprocessors.

\subsection{Comparison to Katzan and Morrison's algorithm.}

To the best of our knowledge, there is only one other abortable RME algorithm,
published recently by Katzan and Morrison \cite{morrison:abrtrme}.
They achieve sublogarithmic complexity:
a process incurs at most $O(\min(k, \log n/\log \log n)$ RMRs in a passage
and $O(f+\min(k, \log n/\log \log n)$ in an attempt.
Furthermore, they achieve these bounds for even the Strict CC multiprocessor.

On the other hand, our work has the following merits.
Unlike the CAS instruction employed in our algorithm,
the fetch\&add instruction, which their algorithm employs to beat 
Attiya et al's lower bound and achieve sublogarithmic complexity,
is not commonly supported by current machines.
Their algorithm does not satisfy FCFS and has a
higher space complexity of $O(n \log^2 n/ \log \log n)$.
Their algorithm is stated to satisfy starvation-freedom if the total number of crashes in the run is finite.
In contrast, our algorithm guarantees that each attempt
completes even in the face of infinitely many crashes in the run, provided that
there are only finitely many crashes during each attempt.

Finally, Katzan and Morrison correctly point out a shortcoming in our conference paper:
our algorithm there admits starvation if there are infinitely many aborts in a run.
The algorithm in this submission has been revised to eliminate this shortcoming.

\subsection{Related Research.}

All of the works on RME prior to the conference version of our paper \cite{jayanti:rmeabort} has focused on designing algorithms that do not provide abortability as a capability.
Golab and Ramaraju \cite{Golab:rmutex} formalized the RME problem and designed several algorithms by adapting traditional mutual exclusion algorithms.
Ramaraju \cite{ramaraju:rglock}, Jayanti and Joshi \cite{jayanti:fcfsmutex}, and Jayanti et al. \cite{jayanti:fasasmutex} 
designed RME algorithms that support the First-Come-First-Served property \cite{Lamport:fcfsmutex}.
Golab and Hendler \cite{Golab:rmutex2} presented an algorithm that has sub-logarithmic RMR complexity on CC machines.
Jayanti et al. \cite{jayanti:rmesublog} presented a unified algorithm that has a sub-logarithmic RMR complexity on both CC and DSM machines.
In another work, Golab and Hendler \cite{Golab:rmutex3} presented an algorithm that has the ideal $O(1)$ passage complexity,
but this result assumes that {\em all} processes in the system crash {\em simultaneously}. 
Recently, Dhoked and Mittal \cite{dhoked:adaptiverme} present an RME algorithm whose RMR complexity adapts to the number of crashes,
and Chan and Woelfel \cite{chan:amortizedrme} present an algorithm which has an O(1) amortized RMR complexity. 
Recently Katzan and Morrison \cite{morrison:abrtrme} gave an abortable RME algorithm that incurs sub-logarithmic RMR on CC and DSM machines.
%For works not on RME but on the theme of crash-restart systems using non-volatile main-memory, 
%Berryhill et al. \cite{berryhill:nvmm}, Izraelevitz et al. \cite{izraelevitz:nvmm}, and 
%Attiya et al. \cite{attiya:rlin} present correctness conditions for recoverable objects.

When it comes to abortability for classical mutual exclusion problem, 
Scott \cite{Scott:abrt} and Scott and Scherer \cite{ScottSch:abrt} designed abortable algorithms
that build on the queue-based algorithms \cite{craig:mcs}\cite{MCS:mutex}.
Jayanti \cite{jayanti:abrt} designed an algorithm based on read, write, and comparison primitives 
having $O(\log n)$ RMR complexity which is also optimal \cite{Attiya:lbound}.
Lee \cite{lee:abrt} designed an algorithm for CC machines that uses the Fetch-and-Add and Fetch-and-Store primitives.
Alon and Morrison \cite{alon:abrt} designed an algorithm for CC machines that has a sub-logarithmic RMR complexity
and uses the read, write, Fetch-And-Store, and comparison primitives.
Recently, Jayanti and Jayanti \cite{jayanti:swapabortable} designed an algorithm for the CC and DSM machines that has a constant amortized RMR complexity 
and uses the read, write, and Fetch-And-Store primitives.
While the works mentioned so far have been deterministic algorithms, 
randomized versions of classical mutual exclusion with abortability exist.
Pareek and Woelfel \cite{pareek:abrt} give a sublogarithmic RMR complexity randomized algorithm
and Giakkoupis and Woelfel \cite{giakkoupis:abrt} give an $O(1)$ expected amortized RMR complexity randomized algorithm.

\section{Specification of the problem} \label{ch2:spec}

In this section, we rigorously specify the Abortable RME problem by defining
what an abortable RME algorithm is,
modeling the algorithm's runs, and
stating the properties that these runs must satisfy.

\subsection{Abortable RME algorithm}

		An {\bf Abortable Recoverable Mutual Exclusion algorithm}, abbreviated {\bf Abortable RME algorithm}, 
		is a tuple $(\procset, \varset, \valset, \initmap, \operations, \partition, \methods)$, where
		
		\begin{itemize}
			\item
			$\procset$ is a set of processes.
			Each process $p \in \procset$ has a set of registers,
			including a {\em program counter}, denoted $\pcp$,
			which points to an instruction in $p$'s code.
			
			\item
			$\varset$ is a set of variables,
			which includes a Boolean variable $\abortsigp$, for each $p \in \cal P$.
			No process except $p$ can invoke any operation on $\abortsigp$,
			and $p$ can only invoke a read operation on $\abortsigp$.
			
			Intuitively, the ``environment''  sets  $\abortsigp$ to $\true$
			when it wishes to communicate to $p$ that it should abort its attempt to acquire the CS and return to the Remainder.
			
			\item
			$\valset$ is a set of values (that each variable in $\varset$ can possibly take on).
			For example, on a 64-bit machine, $\valset$ would be the set of all 64-bit integers. 
			
			\item
			$\initmap$ is a function that assigns a value from $\valset$ to each variable in $\varset$.
			For all $X \in \varset$, $\initmap(X)$ is $X$'s {\em initial value}.
			
			\item
			$\operations$ is a set of operations that each variable in $\varset - \{\abortsigp \mid p \in \procset \}$ supports.
			
			For the algorithm in this paper, $\operations= \{\mbox{{\em read}, {\em write}, {\em CAS}}\}$, where 
				\cas$(X, r, s)$, when executed by a process $p$ (and
$X$ is a variable and $r$, $s$ are $p$'s registers), compares the values of $X$ and $r$; 
				if they are equal, the operation writes in $X$ the value in $s$ and returns $\true$;
				otherwise, the operation returns $\false$, leaving $X$ unchanged.
			
			\item
			$\partition$ is a partition of $\varset$ into $|\procset|$ sets, named $\partition(p)$, for each $p \in \procset$.
			Intuitively, $\partition(p)$ is the set of variables that reside locally at process $p$'s partition on a DSM machine, but has no relevance on a CC machine.
			
			\item
			$\methods$ is a set of methods, which includes 
			three methods per process $p \in \procset$, named $\tryprocp()$, $\exitprocp()$, and $\recoverprocp()$, such that:
			\begin{itemize}
				\item
				In any instruction of any method, 
				at most one operation is performed and 
				it is performed on a single variable from $\varset$.
				\item
				The methods $\tryprocp()$ and $\recoverprocp()$ return a value from $\{\gotocs, \gotorem\}$, and
				$\exitprocp()$ has no return value.
				
				\item
				None of $\tryprocp()$, $\exitprocp()$, or $\recoverprocp()$ calls itself or the other two.
				(This assumption simplifies the model, but is not limiting in any way because 
				it does not preclude the use of helper methods each of which can call itself or the other helper methods.)
			\end{itemize}
		\end{itemize}

\subsection{Abstract sections of code and abstract variables}

For each process $p \in \procset$, we model $p$'s code outside of the methods in $\methods$
to consist of two disjoint sections, named $\remprocp()$ and $\critsecp()$.
Furthermore, we introduce the following {\em abstract} variables,
which are not in $\varset$ and not accessed by the methods in $\methods$,
but are helpful in defining the problem.

\begin{itemize}	
	\item
	$\status_p \in \{\good, \rectry, $\reccs, $\recexit, \recrem\}$.

	Informally, $\status_p$ models $p$'s ``recovery status''. If $\status_p \neq \good$, it means that $p$ is still recovering from a crash, and in this case, the value of $\status_p$ reveals the section of code where $p$ most recently crashed.
	
	\item
	$\cache_p$ holds a set of pairs of the form $(X, v)$,
	where $X \in \varset$ and $v \in \valset$.
	Informally, if $(X,v)$ is present in the cache, $X$ is in $p$'s cache and $v$ is its current value.
	This abstract variable helps define what operations count as {\em remote memory references} (RMR)  on
	CC machines.
\end{itemize}

\subsection{Run, Fair Run, Passage, Attempt} \label{ch2:executionmodel}

A {\bf state} of a process $p$ is a function that assigns a value to each of $p$'s registers, 
including $\pcp$, and a value to each of $\status_p$, $\abortsigp$, and $\cache_p$.

A {\bf configuration} is a function that assigns a state to each process in $\procset$ 
and a value to each variable in $\varset$.
(Intuitively, a configuration is a snapshot of the states of processes and values of variables at a point in time.)

An {\bf initial configuration} is a configuration where, for each $p \in \procset$, 
$\pcp = \remprocp()$, $\status_p = \good$, $\abortsigp = \false$, and $\cache_p = \emptyset$;
and, for each $X \in \varset$, $X = \initmap(X)$.

A {\bf run} is a finite sequence 
$C_0, \alpha_1, C_1, \alpha_2, C_2,  \ldots \alpha_{k}, C_k$, or an infinite sequence \linebreak
$C_0, \alpha_1, C_1, \alpha_2, C_2, \ldots$ such that:
		
\begin{enumerate}
\item
$C_0$ is an initial configuration and, for each $i$, $C_i$ is a configuration and 
$\alpha_i$ is either $(p, \normal)$ or $(p, \crash)$, for some $p \in \procset$.

We call each triple $(C_{i-1}, \alpha_i, C_i)$ a {\em step};
it is a {\em normal step of $p$} if $\alpha_i = (p, \normal)$, and a {\em crash step of $p$} if $\alpha_i = (p, \crash)$.
			
\item
For each normal step $(C_{i-1}, (p, \normal), C_i)$, $C_i$ is the configuration that results
when $p$ executes an enabled instruction of its code, explained as follows:
			
\begin{itemize}
\item
If $\pcp = \remprocp()$ and $\status_p = \good$ in $C_{i-1}$, 
then $p$ invokes either $\tryprocp()$ or $\recoverprocp()$.
				
\item
If $\pcp = \remprocp()$ and $\status_p \neq \good$ in $C_{i-1}$, then $p$ invokes $\recoverprocp()$.
				
\item
If $\pcp = \critsecp()$, then $p$ invokes $\exitprocp()$.

\item
Otherwise, $p$ executes the instruction that $\pcp$ points to in $C_{i-1}$. \\
If this instruction returns $\gotocs$ (resp., $\gotorem$), 
$\pcp$ is set to $\critsecp()$ (resp., $\remprocp()$). \\
If the instruction causes $p$ to return from $\recoverprocp()$, $\status_p$ is set to $\good$ in $C_i$. \\
If $p$ performs a read on $X$ and $X$ is not present in $\cache_p$ in $C_{i-1}$,
then $(X, v)$ is inserted in $\cache_p$, where $v$ is $X$'s value in $C_{i-1}$. \\
In the Strict-CC model, if $p$ performs a non-read operation on $X$,
$X$ is removed from $\cache_q$, for all $q \in \procset$. \\
In the Relaxed-CC model, if $p$ performs a non-read operation on $X$ that changes $X$'s value,
$X$ is removed from $\cache_q$, for all $q \in \procset$.
\end{itemize}
			
\item
For each crash step $(C_{i-1}, (p, \crash), C_i)$, we have:
			
\begin{itemize}
\item
In $C_i$, $\pcp$ is set to $\remprocp()$ and all other registers of $p$ are set to arbitrary values,
and $\cache_p$ is set to $\emptyset$.
				
\item
If $\status_p \neq \good$ in $C_{i-1}$, then $\status_p$ remains unchanged in $C_i$.
Otherwise, if (in $C_{i-1}$) $p$ is in $\tryprocp()$ (respectively, $\critsecp()$, $\exitprocp()$, or $\recoverprocp()$),
then $\status_p$ is set in $C_i$ to $\rectry$ (respectively, $\reccs$, $\recexit$, or $\recrem$).
\end{itemize}
			
\end{enumerate}

A run $R = C_0, \alpha_1, C_1, \alpha_2, C_2,  \ldots$ is {\bf fair} if and only if either $R$ is finite or, 
for all configurations $C_i$ and for all processes $p \in \procset$, the following condition is satisfied:
unless $\pcp = \remprocp()$ and $\status_p = \good$ in $C_i$, 
$p$ has a step in the suffix of $R$ from $C_i$.

Thus, in a fair run, a crashed process eventually restarts,
no process stays in the CS forever, and 
no process permanently ceases to take steps when it is outside the Remainder section.

A {\bf passage} of a process $p$ is a contiguous sequence $\sigma$ of steps in a run such that
$p$ leaves $\remprocp()$ in the first step of $\sigma$ and the last step of $\sigma$ is the earliest subsequent step in the run where $p$ reenters $\remprocp()$ (either because $p$ crashes or because $p$'s method returns $\gotorem$).

An {\bf attempt} of a process $p$ is a maximal contiguous sequence $\sigma$ of steps in a run such that
$p$ leaves $\remprocp()$ in the first step of $\sigma$ with $\status_p = \good$ and the last step of $\sigma$ is the earliest subsequent normal step in the run that causes $p$ to reenter $\remprocp()$ (which would be a return from $\exitprocp$, or a return of $\gotorem$ from $\tryprocp$ or $\recoverprocp$).

\subsection{Remote Memory Reference (RMR) and Point Contention}

A step of $p$ is an {\bf RMR on a DSM machine} if and only if it is a normal step in which
$p$ performs an operation on some variable that is not in $\partition(p)$.

A step of $p$ is an {\bf RMR on a Strict or Relaxed CC machine} if and only if it is a normal step in which
$p$ performs a non-read operation, or $p$ reads some variable that is not present in $p$'s cache.

The {\bf point contention} at a configuration $C$ is the number of processes $p$ such that
$(\pcp \neq \remprocp) \vee (\status_p \neq \good)$ in $C$.

\subsection{Desirable properties} \label{ch2:properties}

We now state the desirable properties of an abortable RME algorithm,
which we divide into three groups---general, recovery-related, and abort-related.

\vspace{0.1in}
\noindent
{\bf General properties}:

		\begin{itemize}
			\item[P1]
			\underline{Mutual Exclusion}: At most one process is in the CS in any configuration of any run.
			
			\item[P2]
			\underline{Bounded Exit}: There is an integer $b$ such that if in any run any process $p$ invokes
			and executes $\exitprocp()$ without crashing, the method completes in at most $b$ steps of $p$.
			
			\item[P3]
			\underline{Weak Starvation Freedom (WSF)}: In every fair infinite run in which there are only finitely many crash steps,
			if a process $p$ is in the Try section in a configuration, $p$ is in a different section in a later configuration.
			\item[P4]
			\underline{Starvation Freedom (SF)}: In every fair infinite run in which every attempt contains 
			only finitely many crash steps,
			if a process $p$ is in the Try section in a configuration, $p$ is in a different section in a later configuration.
			%\end{itemize}
			
			We note that SF implies WSF.
			
			\item[P5]
			\underline{First-Come-First-Served (FCFS)}:
	There is an integer $b$ such that in any run,
	if $A$ and $A'$ are attempts by any distinct processes $p$ and $p'$, respectively,
	$p$ performs at least $b$ consecutive normal steps in $A$ before the attempt $A'$ starts, 
	and $p$ neither receives an abort signal nor subsequently crashes in $\tryprocp()$ in $A$,
	then $p'$ does not enter the CS in $A'$ before $p$ enters the CS in $A$.
			
		\end{itemize}

\vspace{0.1in}
\noindent
{\bf Recovery related properties}:
		
		\begin{itemize}
			\item[P6]
			\underline{Critical Section Reentry (CSR) \cite{Golab:rmutex}}:
			In any run, if a process $p$ crashes while in the CS, 
			no other process enters the CS until $p$ subsequently reenters the CS.
			
			\item[P7]
			\underline{Bounded Recovery to CS}:
			There is an integer $b$ such that if in any run any process $p$ 
			executes $\recoverprocp()$ without crashing and with $\status_p = \reccs$, 
			the method completes in at most $b$ steps of $p$ and returns $\gotocs$.
			
			\item[P8]
			\underline{Bounded Recovery to Exit}:
			There is an integer $b$ such that if in any run any process $p$ 
			executes $\recoverprocp()$ without crashing and with $\status_p = \recexit$, 
			the method completes in at most $b$ steps of $p$.
			
			\item[P9]
			\underline{Fast Recovery to Remainder}:
			There is an absolute constant $b$, i.e., a constant independent of $|\procset|$, such that if in any run any process $p$ 
			executes $\recoverprocp()$ without crashing and with $\status_p \in \{\good, \recrem\}$, 
			the method completes in at most $b$ steps of $p$.
			
			\item[P10]
			\underline{Bounded Recovery to Remainder}:
			There is an integer $b$ such that if in any run
			$\recoverprocp()$, executed by a process $p$ with $\status_p = \rectry$, returns $\gotorem$,
			$p$ must have completed that execution of $\recoverprocp()$ in at most $b$ of its steps.
		\end{itemize}

\vspace{0.1in}
\noindent
{\bf Abort related properties}:

\vspace{0.08in}
\noindent
For any run $R$, any configuration $C$ of $R$, and any process $p$,
define the predicate $\beta(R, p, C)$ as $\true$ if and only if in configuration $C$
it is the case that $\abortsigp$ is $\true$ and
either $p$ is in $\tryprocp()$ or $p$ is in $\recoverprocp()$ with $\status_p = \rectry$.

		\begin{itemize}
			\item[P11]
			\underline{Bounded Abort}: 
			There is an integer $b$ such that, for each $R, C, p$,
			if $\beta(R, p, C)$ is true,
			$\abortsigp$ stays $\true$ for ever (i.e., stays true in the suffix $R'$ of the run from $C$), 
			and $p$ executes steps without crashing (i.e., $p$ has no crash steps in $R'$),
			then $p$ enters either the CS or the remainder in at most $b$ of its steps (in $R'$).
			
			\item[P12]
			\underline{No Trivial Aborts}: 
			In any run, if $\abortsigp$ is $\false$ when a process $p$ invokes $\tryprocp()$,
			$\abortsigp$ remains $\false$ forever, and $p$ executes steps without crashing,
			then $\tryprocp()$ does not return $\gotorem$.
		\end{itemize}

\section{The Algorithm} \label{sec:algo}
We present our abortable RME algorithm in Figure~\ref{algo:abrt}. 
The algorithm is designed for the set of processes $\procset =  \{ 1, 2, \ldots, n \}$.
All the shared variables used by our algorithm are stored in NVRAM.
Variables with a subscript of $p$ to their name are local to process $p$, and are stored in $p$'s registers.
%We assume that the external signal to process $p$ asking it to abort is made available at $\abort[p]$ as a boolean value,
%with a value of $\true$ indicating that the signal is active and a value of $\false$ indicating otherwise.

\begin{figure}
	\begin{footnotesize}
		%\hrule
		\vspace{0mm}
		\begin{tabbing}
			\hspace{0in} \= {\bf Persistent variables (stored in NVRAM)} \hspace{0.2in} \= \hspace{0.2in} \=  \hspace{0.2in} \= \hspace{0.2in} \= \hspace{0.2in} \=\\
			\> \hspace{0.1in} \= $\minarray[1 \dots |\procset|]$ : A min-array; initially $\minarray[p] = (p, \infty)$, for all $p \in \procset$.\\
			\> \> $\csowner \in \{0\} \times (\{ 0 \} \cup  \mathbb{N}^{+}) \cup \{1\} \times \procset$; initially $(0, 1)$. \\
			\> \> $\seq \in \mathbb{N}$; initially 1. \\
			\> \> $\forall p  \in \procset, \go[p] \in \mathbb{N}^{+} \cup \{-1, 0\}$, initially $\perp$. \\
			\> \> $\token \in \mathbb{N}$, initially 1. 
		\end{tabbing}
		\vspace*{-6mm}
		\begin{tabbing}
			\hspace{0in} \= \hspace{0.2in} \= \hspace{0.2in} \=  \hspace{0.2in} \= \hspace{0.2in} \= \hspace{0.2in} \=\\
			\> \procline{abrt:rem:1} \> {\texttt{Remainder Section}} \\
			\\
			\> \> \underline{{\bf procedure} $\tryprocp()$:} \\
			\> \procline{abrt:try:1} \> $\tokp \leftarrow \token$ \\
			\> \procline{abrt:try:2} \> \cas$(\token, \tokp, \tokp + 1)$ \\
			\> \procline{abrt:try:3} \> $\go[p] \leftarrow \tokp$ \\
			\> \procline{abrt:try:4} \> $\registry[p].\mwrite((p, \tokp))$ \\
			\> \procline{abrt:try:5} \> $\promotep(\false)$ \\
			\> \procline{abrt:try:6} \> \waittill $\go[p] = 0 \vee \abort[p]$ \\
			\> \procline{abrt:try:7} \> \ifcode $\go[p] = 0$: \return $\gotocs$ \\
			\> \procline{abrt:try:8} \> \return $\abortprocp()$ \\
			\\
			\> \procline{abrt:cs:1} \> {\texttt{Critical Section}} \\
			\\
			\> \> \underline{{\bf procedure} $\exitprocp()$:} \\
			\> \procline{abrt:exit:1} \> $\registry[p].\mwrite((p, \infty))$ \\
			\> \procline{abrt:exit:2} \> $\myseqp \leftarrow \seq$ \\
			\> \procline{abrt:exit:3} \> $\seq \leftarrow \myseqp+1$ \\
			\> \procline{abrt:exit:4} \> $\csowner \leftarrow (0, \myseqp+1)$ \\
			\> \procline{abrt:exit:5} \> $\promotep(\false)$ \\
			\> \procline{abrt:exit:6} \> $\go[p] \gets -1$ \\
			\\
			\> \> \underline{{\bf procedure} $\recoverprocp()$:} \\
			\> \procline{abrt:rec:1} \> \ifcode $\go[p] = -1$: \return $\gotorem$ \\
			\> \procline{abrt:rec:2} \> \return $\abortprocp()$\\
			\\ 
			\> \> \underline{{\bf procedure} $\abortprocp()$:} \\
			\> \procline{abrt:abort:1} \> $\registry[p].\mwrite((p, \infty))$ \\
			\> \procline{abrt:abort:2} \> $\promotep(\true)$ \\
			\> \procline{abrt:abort:3} \> \ifcode $\csowner = (1, p)$: \return $\gotocs$ \\
			\> \procline{abrt:abort:4} \> $\go[p] \gets -1$ \\
			\> \procline{abrt:abort:5} \> \return $\gotorem$\\
			\\ 
			\> \> \underline{{\bf procedure} $\promotep({\bf boolean} \ \isabortingp)$:} \\
			\> \procline{abrt:prom:1} \> $(\bitp, \myseqp) \leftarrow \csowner$; \ifcode $\bitp = 1$: \{ $\peerp \gets \myseqp$; \goto Line~\refln{abrt:prom:4} \} \\
			\> \procline{abrt:prom:2} \> $(\peerp, \tokp) \leftarrow \registry.\findmin()$; \ifcode $\tokp = \infty \wedge \isabortingp$: $\peerp \gets p$ \elsecode \ifcode $\tokp = \infty$: \return \\
			%\> \hspace*{1mm}$\lfloor$ \> \>  \ifcode $\tokp = \infty \wedge \isabortingp$: $\peerp \gets p$ \elsecode \ifcode $\tokp = \infty$: \return \\
			\> \procline{abrt:prom:3} \> \ifcode $\neg \cas(\csowner, (0, \myseqp), (1, \peerp))$: \return \\
			\> \procline{abrt:prom:4} \> $\mygop \leftarrow \go[\peerp]$; \ifcode $\mygop \in \{-1, 0 \}$: \return  \\
			\> \procline{abrt:prom:5} \> \ifcode $\csowner \neq (1, \peerp)$: \return \\
			\> \procline{abrt:prom:6} \> \cas$(\go[\peerp], \mygop, 0)$ 
		\end{tabbing}
		
		\caption{Abortable RME Algorithm for CC and DSM machines. Code for process $p$.}
		\label{algo:abrt} % TIP: label always appears after caption
	\end{footnotesize}
	\hrule
\end{figure}

\subsection{Shared variables and their purpose} \label{sec:sharedvars}
We describe below the role played by each shared variable used in the algorithm.

\begin{itemize}
\item
$\token$ is an unbounded positive integer.
A process $p$ reads this variable at the beginning of $\tryprocp()$ to obtain its token
and then increments, thereby ensuring that processes that invoke the try method later
will get a strictly bigger token.

\vspace{0.1in}
\item
$\csowner$ and $\seq$:
These two shared variables are used in conjunction, with
$\seq$ holding an unbounded integer and
$\csowner$ holding a pair, which is either $(\true, p)$ (for some $p \in \procset$) or $(\false, \seq)$.
If $\csowner = (\true, p)$, it means that $p$ owns the CS and, if $\csowner = (\false, \seq)$, it means that 
no process owns the CS.
If $\seq$ has a value $s$ while $p$ is the CS,
when exiting the CS $p$ increments $\seq$ to $s+1$ and writes $(0, s+1)$ in $\csowner$.
As we explain later, this act is crucial to ensuring that no process will be made the owner of the CS
after it has moved back to the remainder.

\vspace{0.1in}
\item
$\go[p]$ has one of three values --- $-1$, 0, or $p$'s token.
The algorithm ensures that $\go[p] = -1$ whenever $p$ is in the remainder ``normally'',
i.e., not because of a crash but because the try, exit, or recover method returned normally.
If $\go[p]=0$, it means that $p$ is made the owner of CS, hence $p$ has the permission to enter the CS.
After $p$ obtains a token in $\tryprocp()$, $p$ writes its token in $\go[p]$ 
and, subsequently when $p$ must wait for its turn to enter the CS,
it spins until either $\go[p]$ turns 0 or it receives a signal to abort.

\vspace{0.1in}
\item
$\minarray$ is a min-array object \cite{farrays} of $n$ locations that supports two operations:
\linebreak
$\registry[p].\mwrite(v)$, which can only be executed by process $p$, writes $v$ in $\registry[p]$;
and $\registry.\findmin()$ returns the minimum value in the array. 
After $p$ obtains a token $t$ in $\tryprocp()$, it announces its interest to capture the CS
by writing the pair $(p, t)$ in $\registry[p]$,
and when no longer interested, it takes itself out by writing $(p, \infty)$ in $\registry[p]$.
The ``less than'' relation on pairs is defined as follows:
$(p, t) < (p', t')$ if and only if $t < t'$ or $(t = t') \wedge (p < p')$.

\vspace{0.05in}
It turns out that the $\registry$ object has an implementation, using only read, write, and CAS operations,
with three nice properties  \cite{farrays}: it is linearizable, wait-free, and idempotent, i.e.,
if $p$ crashes while executing the method $\registry[p].\mwrite(v)$ and 
reexecutes the method once more upon restart, the effect is the same as executing the method once without ever crashing.
The implementation uses only $O(n)$ variables
and has only a logarithmic RMR complexity on a DSM or a Relaxed CC machine:
$\registry.\findmin()$ incurs $O(1)$ RMRs and $\registry[p].\mwrite(v)$ incurs $O(\min(k, \log n))$ RMRs,
where $k$ is the maximum point contention during the execution of $\registry[p].\mwrite(v)$.
The idempotence property of the implementation makes it suitable for use in our algorithm \cite{jayanti:fcfsmutex}.
\end{itemize}

\subsection{Informal description}

In this section we present an intuitive understanding of the algorithm
that explains the lines of code and, more importantly, draws attention to
potential race conditions and how the algorithm overcomes them.

\vspace{0.1in}
\noindent
{\bf Understanding $\tryprocp()$}

After a process $p$ invokes $\tryprocp()$, it reads and then attempts to increments $\token$
(Lines~\refln{abrt:try:1}, \refln{abrt:try:2}).
The attempt to increment serves two purposes.
First, if a different process $q$ invokes $\tryproc{q}()$ later,
it gets a strictly larger token, which helps realize FCFS.
Second, if $p$ were to abort its curent attempt $A$, it will obtain a strictly larger token in its next attempt $A'$,
which, as we will see, helps ensure that any process $q$ that might attempt to release $p$ from its busy-wait 
in the attempt $A$ will not accidentally release $p$ from its busy-wait in the attempt $A'$.
Process $p$ writes its token in $\go[p]$ (Line~\refln{abrt:try:3}), 
where it will later busy-wait until some process changes $\go[p]$ to 0, and 
then announces its interest in the CS by changing $\registry[p]$ from
$(p, \infty)$ to $(p, \mbox{its token})$ (Line~\refln{abrt:try:4}).
It then calls the $\promotep()$ procedure, which is crucial to ensuring livelock-freedom (Line~\refln{abrt:try:5}).

\vspace{0.1in}
\noindent
{\bf Understanding $\promotep()$}

The $\promotep()$ procedure's purpose is to push a waiting process into the CS, if the CS is unoccupied.
To this end, $p$ reads $\csowner$ (Line~\refln{abrt:prom:1}).
If it finds that the CS is already owned (i.e., $b_p = 1$), since it is possible that the owner 
$peer_p$ is still busywaiting unaware of its ownership,
$p$ jumps to Line~\refln{abrt:prom:4}, where the code to release $peer_p$ starts.
On the other hand, if the CS is unoccupied (i.e., $b_p = 0$), it executes Line~\refln{abrt:prom:2} 
to find out the process
that has the smallest token in the $\registry$, i.e., the process $peer_p$ that has been waiting the longest.
Since $\promotep()$ is called from $p$'s Line~\refln{abrt:try:5},
at which point $\registry[p]$ has a finite token number for $p$,
at Line~\refln{abrt:prom:2} we have $\tokp \neq \infty$.
So, $p$ proceeds to Line~\refln{abrt:prom:3}, where it attempts to launch $peer_p$ into the CS.
If $p$'s CAS fails, it means that someone else must have succeeded in launching
a process into the CS between $p$'s Line~\refln{abrt:prom:1} and Line~\refln{abrt:prom:3};
in this case $p$ has no further role to play, so it returns from the procedure.
On the other hand, if $p$'s CAS succeeds, which means that
$peer_p$ has been made the CS owner, $p$ has a responsibility to release $peer_p$ from its busywait,
i.e., $p$ must write 0 in $\go[peer_p]$.
However, there is potential for a nasty race condition here,
as explained by the following scenario:
some process different from $p$ releases $peer_p$ from its busywait;
$peer_p$ enters the CS and then exits to the remainder;
some other process $q$ is now in the CS;
$peer_p$ executes the try method once more and proceeds up to the point of busy-waiting.
Recall that $p$ is poised to write 0 in $\go[peer_p]$.
If $p$ does this writing, $peer_p$ will be released from its busywait,
so $peer_p$ proceeds to the CS, where $q$ is already present.
So, mutual exclusion is violated!
Our algorithm averts this disaster by exploiting the fact that, while $peer_p$ busywaits,
$\go[peer_p]$'s value is never the same between different attempts of $peer_p$.
Specifically, $p$ reads $\go[peer_p]$ (Line~\refln{abrt:prom:4});
if $g_p$ is $-1$ or 0, it means that $peer_p$ is not busywaiting, so $p$ has no role to play, hence it returns.
If things have moved on and $peer_p$ no longer owns the CS, then too $p$ has 
no role to play, hence it returns (Line~\refln{abrt:prom:5}).
Otherwise, there are two possibilities: either $\go[peer_p]$ is still $g_p$ or it has changed.
In the former case, $peer_p$ must be busywaiting, 
so it is imperative that $p$ takes the responsibility to release $peer_p$ (by changing $\go[peer_p]$ to 0).
In the latter case, $peer_p$ requires no help from $p$, so $p$ must not change $\go[peer_p]$
(in order to avoid the race condition described above).
This is precisely what the CAS at Line~\refln{abrt:prom:6} accomplishes.

\vspace{0.1in}
\noindent
{\bf The rest of  $\tryprocp()$}

Upon returning from $\promotep()$, $p$ busywaits until it reads a 0 in $\go[p]$ or it receives a request to abort
(Line~\refln{abrt:try:6}).
If $p$ reads a 0 in $\go[p]$, $p$ infers that it owns the CS, so $\tryprocp()$ returns $\gotocs$ (Line~\refln{abrt:try:7}).
If $p$ receives a request to abort, it calls $\abortprocp()$ (Line~\refln{abrt:try:8}),
which we describe next.

\vspace{0.1in}
\noindent
{\bf Understanding $\abortprocp()$}

To abort, $p$ writes $(p, \infty)$ to make it known to all that it has no interest in capturing the CS (Line~\refln{abrt:abort:1}).
If any process will invoke the promote procedure after this point,
it will not find $p$ in $\registry$, so it will not attempt to launch $p$ into the CS.
Does this mean that $p$ can now return to the remainder section?
The answer is a thundering no because there are two nasty race conditions that need to be overcome.

First, it is possible that, before $p$ performed Line~\refln{abrt:abort:1},
some process $q$ performed its Line~\refln{abrt:prom:2} to find $p$ in $\registry$,
and then successfully launched $p$ into the CS (by writing $(1,p)$ in $\csowner$).
Taking care of this scenario is easy: $p$ can read $\csowner$ and if $p$ finds that it owns the CS,
it can abort by simply returning $\gotocs$.

The second potential race is more subtle and harder to overcome.
As in the earlier scenario, suppose that, before $p$ performed Line~\refln{abrt:abort:1},
some process $q$ performed its Line~\refln{abrt:prom:2} to find $p$ in $\registry$
(i.e., $peer_q = p$).
Furthermore, suppose that $q$ is now at Line~\refln{abrt:prom:3} and $\csowner = (0, s_q)$.
So, after performing  Line~\refln{abrt:abort:1}, if $p$ naively returns to the remainder
and then $q$ performs Line~\refln{abrt:prom:3}, we would be in a situation
where $p$ has been made the CS owner after it was back in the remainder!

To overcome the above two race conditions, $p$ calls $\promotep(\true)$ (Line~\refln{abrt:abort:2}).

The parameter $\true$ conveys that the call is made by $p$ while aborting,
and has the following impact on how $p$ executes $\promotep()$:
if $p$ finds the CS to be unoccupied at Line~\refln{abrt:prom:1}
and finds $\registry$ to be empty at Line~\refln{abrt:prom:2},
to preempt the second race condition discussed above
(where some process $q$ is poised to launch $p$ into the CS),
$p$ will attempt to launch itself into the CS 
(by setting $peer_p$ to $p$ at Line~\refln{abrt:prom:2} and attempting to change
$\csowner$ to $(1, peer_p)$).
The key insight is that, after $p$ performs the CAS at Line~\refln{abrt:prom:3},
only two possibilities remain: either $p$ is already launched into the CS (i.e., $\csowner = (1,p)$) or
it is guaranteed that no process will launch $p$ into the CS.
In the former case, $\abortprocp()$ returns $\gotocs$ at Line~\refln{abrt:abort:3};
and in the latter case, since it is safe for $p$ to return to the remainder,
$\abortprocp()$ returns $\gotorem$ at Line~\refln{abrt:abort:5} after setting $\go[p]$ to $-1$
at Line~\refln{abrt:abort:4} (in order to respect the earlier mentioned invariant
that $\go[p] = -1$ whenever $p$ returns to the remainder normally).

\vspace{0.1in}
\noindent
{\bf Understanding $\exitprocp()$}

There are two routes by which $p$ might enter the CS.
One is the ``normal'' route where $p$ executes $\tryprocp()$ without aborting or crashing,
and $\tryprocp()$ returns $\gotocs$, thereby sending $p$ to the CS.
The second route is where $p$ receives an abort signal,
calls at Line~\refln{abrt:try:8} $\abortprocp()$, which returns $\gotocs$ at Line~\refln{abrt:abort:3},
causing $\tryprocp()$ also to return $\gotocs$ at Line~\refln{abrt:try:8}.
When $p$ is in the CS, $p$'s announcement in $\registry[p]$ (made at Line~\refln{abrt:try:4}),
would no longer be there if it entered the CS by the second route (because of Line~\refln{abrt:abort:1}),
but it would still be there if it entered the CS by the first route.
So, when $p$ exits the CS, it removes its announcement in $\registry[p]$ (Line~\refln{abrt:exit:1}).
It then increments the number in $\seq$
and gives up its ownership of the CS by changing $\csowner$ from $(1,p)$ to $(0, \seq)$
(Lines~\refln{abrt:exit:2},~\refln{abrt:exit:3},~\refln{abrt:exit:4}).
To launch a waiting process, if any, into the just vacated CS,
$p$ then executes $\promotep()$ (Line~\refln{abrt:exit:5}),
and returns to the remainder after setting $\go[p]$ to $-1$
at Line~\refln{abrt:exit:6} (in order to respect the earlier mentioned invariant
that $\go[p] = -1$ whenever $p$ returns to the remainder normally).

\vspace{0.1in}
\noindent
{\bf Understanding $\recoverprocp()$}

Process $p$ executes $\recoverprocp()$ when it restarts after a crash.
If $\go[p]$ has $-1$, $p$ infers that
either $\recoverprocp()$ was called when $\status_p = \good$ or
the most recent crash had occured early in $\tryprocp()$, so 
$\recoverprocp()$ simply sends $p$ back to the remainder (Line~\refln{abrt:rec:1}).
Otherwise, $\recoverprocp()$ simply calls $\abortprocp()$ (Line~\refln{abrt:rec:1}),
which does the needful.
In particular, if $p$ was in the CS at the most recent crash, then
$\csowner$ would have $(1,p)$, which causes $\abortprocp()$ to send $p$ back to the CS.
Otherwise, $\abortprocp()$ extricates $p$ from the algorithm, sending it either to the CS or to the remainder.

\section{Proof of Correctness}

Figure~\ref{inv:abrt} presents the invariant satisfied by the Abortable RME algorithm given in Figure~\ref{algo:abrt}.
We have written the 13 statements comprising the invariant with the following conventions.
All statements about process $p$ are universally quantified, i.e., $\forall p \in \procset$ is implicit
(these are Statements 3 through 11, and Statement 13).
The program counter for a process $p$, i.e., $\pcp$, can take any of the values from the set $[\refln{abrt:rem:1}, \refln{abrt:prom:6}]$.
However, when a call to procedure $\promotep()$ is made by $p$ and $p$ is executing one of the steps from Lines~\refln{abrt:prom:1}-\refln{abrt:prom:6}, 
for clearly conveying where the call was made from, we prefix the value of $\pcp$ with the line number from where $\promotep()$ was called, along with the scope resolution operator from C++, namely, ``::''.
Thus, $\pcp = \resolve{abrt:try:5}{abrt:prom:4}$ means $p$ called $\promotep()$ from Line~\refln{abrt:try:5} and is now executing Line~\refln{abrt:prom:4} in that call.
Sometimes, in the interest of brevity, we use the range operator, i.e., $[a, b]$, to convey something more than just saying the range of values from $a$ to $b$ (inclusive).
That is, if $\pcp \in [\refln{abrt:try:5}, \refln{abrt:try:7}]$, we also mean that $\pcp$ could take on values from $[\resolve{abrt:try:5}{abrt:prom:1}, \resolve{abrt:try:5}{abrt:prom:6}]$ because there is a call to $\promotep()$ at Line~\refln{abrt:try:5}.
Similarly, $\pcp \in [\refln{abrt:try:4}, \refln{abrt:try:5}]$ means that $\pcp$ takes on values from $[\resolve{abrt:try:5}{abrt:prom:1}, \resolve{abrt:try:5}{abrt:prom:6}]$ because, again, there is a call to $\promotep()$ at Line~\refln{abrt:try:5}.

\begin{figure}[!h]
	%	\newgeometry{left=3cm} %{top=5mm, bottom=10mm} %{left=3cm,bottom=0.1cm}
	\hrule
	\begin{footnotesize}
		\hrule
		\vspace{1mm}
		\begin{tabbing}
			\hspace{0in} \= {\bf Conditions:} \hspace{0.2in} \= \hspace{0.2in} \=  \hspace{0.2in} \= \hspace{0.2in} \= \hspace{0.2in} \=\\
		\end{tabbing}
		\vspace{-0.25in}
		\begin{enumerate}
			\item \cond{1} $\token \ge 1$
			
			\item \cond{2} $(\csowner = (0, \seq)) \, \vee \, (\exists q \in \procset, \, \csowner = (1,q))$
			
			\item \cond{5} $(-1 \le \go[p] < \token)$ $\wedge$ $(\pcp = \refln{abrt:try:4} \limplies \go[p] = \tokp)$ $\wedge$ $(\pcp \in [ \refln{abrt:try:5}, \refln{abrt:try:7} ] \limplies \go[p] \in \{0, \tokp\})$  \\
			\hspace*{3mm} $\wedge$ $(\pcp \in \{ \refln{abrt:try:8} \upto \refln{abrt:exit:6}, \refln{abrt:rec:2} \upto \refln{abrt:abort:4}, \refln{abrt:prom:1} \upto \refln{abrt:prom:6} \} \limplies \go[p] \neq -1)$ \\
			\hspace*{3mm} $\wedge$ $((\pcp \in \{ \refln{abrt:try:1} \upto \refln{abrt:try:3}, \refln{abrt:abort:5}  \} \vee (\pcp \in \{ \refln{abrt:rem:1}, \refln{abrt:rec:1} \} \wedge \status_p \in \{ \good, \recrem \})) \limplies \go[p] = -1)$  
			
			\item \cond{6} $( \exists t \in [1, \token-1] \cup \{\infty\}, \minarray[p] = (p,t))$\\ 
			\hspace*{3mm}  $\wedge$ $(\pcp \in [ \refln{abrt:try:5}, \refln{abrt:try:7} ] \limplies \minarray[p] = (p, \tokp))$\\
			\hspace*{3mm}  $\wedge$ $((\pcp \in \{\refln{abrt:try:4}, \refln{abrt:exit:2} \upto \refln{abrt:exit:6}, \refln{abrt:abort:2} \upto \refln{abrt:abort:4} \} \vee  \go[p] = -1)$ $\limplies$ 
			$\minarray[p] = (p, \infty))$
			
			\item \cond{7} $(((\pcp \in [\refln{abrt:try:5}, \refln{abrt:try:7}]$ $\wedge$ $\go[p] = 0) \vee \pcp \in [ \refln{abrt:cs:1}, \refln{abrt:exit:4} ] \vee \status_p = \reccs) \limplies \csowner = (1, p))$ \\
			\hspace*{3mm}  $\wedge$ $((\pcp \in  \{ \refln{abrt:try:4}, \refln{abrt:abort:4} \} \cup [ \refln{abrt:exit:5}, \refln{abrt:exit:6} ]$
			$\vee$ $ \go[p] = -1) $  $\limplies$ $\csowner \neq (1, p))$
			
			\item \cond{8} This condition states what values local variables of process $p$ take on. \\
			$(\pcp = \refln{abrt:try:2} \limplies 1 \leq \tokp \leq \token)$ $\wedge$ $(\pcp \in [\refln{abrt:try:3}, \refln{abrt:try:7}] \limplies 1 \leq \tokp < \token)$ \\
			\hspace*{3mm} $\wedge$ $(\pcp = \refln{abrt:exit:3} \limplies \myseqp = \seq)$ $\wedge$ $(\pcp =\refln{abrt:exit:4} \limplies \myseqp = \seq - 1)$ \\
			\hspace*{3mm} $\wedge$ $(\pcp \in [\resolve{abrt:try:5}{abrt:prom:1}, \resolve{abrt:try:5}{abrt:prom:6}] \cup [\resolve{abrt:exit:5}{abrt:prom:1}, \resolve{abrt:exit:5}{abrt:prom:6}] \limplies \isabortingp = \false)$ $\wedge$ $(\pcp \in [\resolve{abrt:abort:2}{abrt:prom:1}, \resolve{abrt:abort:2}{abrt:prom:6}] \limplies \isabortingp = \true)$\\
			\hspace*{3mm} $\wedge$ $(\pcp \in [\refln{abrt:prom:3}, \refln{abrt:prom:6}] \limplies \peerp \in \procset)$ \\
			\hspace*{3mm} $\wedge$ $(\pcp \in [\refln{abrt:try:1}, \refln{abrt:exit:6}] \limplies \status_p = \good)$
			
			\item \cond{9} $(\pcp = \refln{abrt:try:7} \limplies (\go[p] = 0 \vee \textnormal{abort was requested}))$ $\wedge$ $(\pcp = \refln{abrt:try:8} \limplies \textnormal{abort was requested})$
			
			\item \cond{10} 
			$\pcp \in \{ \refln{abrt:prom:2}, \refln{abrt:prom:3} \} \limplies (\myseqp \leq \seq$ $\wedge$ $(\forall q, \pc{q} \in \{ \refln{abrt:exit:3}, \refln{abrt:exit:4} \} \limplies \myseqp \leq \myseq{q}))$
			
			\item \cond{11} 
			$((\pcp =\refln{abrt:prom:2} \wedge \csowner = (0, \myseqp))$ $\limplies$ \\
			\hspace*{25mm}$\forall q, (\minarray[q] \neq (q, \infty) \limplies ( \pc{q} \in \{ \refln{abrt:try:5} \upto \refln{abrt:try:8}, \refln{abrt:rec:2}, \refln{abrt:abort:1} \} \vee (\pc{q} \in \{ \refln{abrt:rem:1}, \refln{abrt:rec:1} \} \wedge \go[q] \neq -1))))$ \\
			\hspace*{3mm} $\wedge$ $((\pcp = \refln{abrt:prom:3} \wedge \csowner = (0, \myseqp)) \limplies (\pc{\peerp} \in [ \refln{abrt:try:5}, \refln{abrt:try:7}] \cup \{ \refln{abrt:rec:2} \upto \refln{abrt:abort:2}, \resolve{abrt:abort:2}{abrt:prom:1}\}$ \\
			\hspace*{50mm} $\vee$ $(\pc{\peerp} \in \{ \resolve{abrt:abort:2}{abrt:prom:2}, \resolve{abrt:abort:2}{abrt:prom:3} \} \wedge \myseq{\peerp} = \myseqp)$ \\
			\hspace*{50mm} $\vee$ $(\pc{\peerp} \in \{ \refln{abrt:rem:1}, \refln{abrt:rec:1} \} \wedge \go[\peerp] \neq -1 )))$
			
			\item \cond{12} 
			$\pcp = \{ \refln{abrt:prom:5}, \refln{abrt:prom:6} \} \limplies 1 \leq \mygop < \token$ 
			
			\item \cond{13}
			$\pcp = \refln{abrt:prom:6} \limplies ((\pc{\peerp} \in \{ \refln{abrt:try:2}, \refln{abrt:try:3} \} \limplies 1 \leq \mygop < \tok{\peerp})$ \\
			\hspace*{20mm} $\wedge$ $(\pc{\peerp} = \refln{abrt:try:4} \limplies 1 \leq \mygop < \go[\peerp])$ \\
			\hspace*{20mm} $\wedge$ $((\pc{\peerp} \in \{ \refln{abrt:try:5}, \refln{abrt:try:6}, \refln{abrt:try:7} \} \wedge \mygop = \go[\peerp]) \limplies \csowner = (1,\peerp)))$
			
			\item \cond{3} If a process is registered, some $q$ is either in CS or can be counted on to launch a waiting process into CS. \\
			\\
			$\min(\minarray) \neq (*, \infty)$ $\limplies$ $\exists q, (\csowner = (1, q)$\\
			\hspace*{45mm} $\vee$ $(\pc{q} \in \{ \refln{abrt:rem:1}, \refln{abrt:rec:1} \} \wedge  \go[q] \neq -1) \vee \pc{q} \in \{ \refln{abrt:try:5}, \refln{abrt:exit:5}, \refln{abrt:rec:2} \upto \refln{abrt:abort:2}, \refln{abrt:prom:1} \}$ \\
			\hspace*{45mm} $\vee$ $(\pc{q} \in \{ \refln{abrt:prom:2}, \refln{abrt:prom:3} \} \wedge \csowner = (0, \myseq{q})))$
			
			\item \cond{4} If $p$ has the ownership of CS but $\go[p] \neq 0$, then there is some $q$ that can be counted on to set $\go[p]$ to $0$. \\
			\\
			$(\csowner = (1, p) \wedge \go[p] \neq 0)$ $\limplies$ $\exists q, (\pc{q} \in \{ \refln{abrt:rec:2} \upto \refln{abrt:abort:2}, \refln{abrt:prom:1} \}$ $\vee$ $(\pc{q} = \refln{abrt:prom:4} \wedge \peer{q} = p)$ \\
			\hspace*{55mm} $\vee$ $(\pc{q} \in \{ \refln{abrt:prom:5}, \refln{abrt:prom:6} \} \wedge \peer{q} = p \wedge \mygo{q} = \go[p])$ \\
			\hspace*{55mm} $\vee$ $(\pc{q} \in \{ \refln{abrt:rem:1}, \refln{abrt:rec:1} \} \wedge \go[q] \neq -1))$
		\end{enumerate}
		
		%\vspace*{-2mm}
		\caption{Invariant of the Abortable RME Algorithm from Figure~\ref{algo:abrt}.}
		\label{inv:abrt}
		\hrule
		\vspace*{-5mm}
	\end{footnotesize}
	%\restoregeometry
\end{figure}

\begin{lemma}[{\bf Mutual Exclusion}] \label{lem:abrt:mutex}
	At most one process is in the CS in any configuration of any run.
\end{lemma}
\begin{proof}
	Suppose there is a configuration $C$ such that two distinct processes $p$ and $q$ are in the CS,
	i.e., $\pcp  = \pc{q} = \refln{abrt:cs:1}$.
	By Condition~\ref{inv:abrt:cond7}, $\csowner = (1,p)$ and $\csowner = (1, q)$ in $C$, 
	which means $\csowner$ has two different values in the same configuration, a contradiction.
\end{proof}

\begin{lemma}[{\bf Bounded Exit}] \label{lem:abrt:boundedexit}
	There is an integer $b$ such that if in any run any process $p$ invokes
	and executes $\exitprocp()$ without crashing, the method completes in at most $b$ steps of $p$.
\end{lemma}
\begin{proof}
	As explained earlier, the call to $\registry[p].\mwrite()$ at Line~\refln{abrt:exit:1} takes $O(\log n)$ steps.
	From an inspection of the algorithm we see that the rest of the execution of $\exitprocp()$ completes in a constant number of steps (this includes the execution of $\registry.\findmin()$ invoked from a call to $\promotep()$).
	It follows that for a certain constant $c$, the execution of $\exitprocp()$ completes in at most $c \log n$ steps in a run, if $p$ invokes and executes it without crashing.
\end{proof}

\begin{lemma}[{\bf First Come First Served}] \label{lem:abrt:fcfs}
	There is an integer $b$ such that in any run,
	if $A$ and $A'$ are attempts by any distinct processes $p$ and $p'$, respectively,
	$p$ performs at least $b$ consecutive normal steps in $A$ before the attempt $A'$ starts, 
	and $p$ neither receives an abort signal nor subsequently crashes in $\tryprocp()$ in $A$,
	then $p'$ does not enter the CS in $A'$ before $p$ enters the CS in $A$.
\end{lemma}
\begin{proof}
	Let $B$ be the earliest configuration when $p$ has performed contiguous normal steps upto Line~\refln{abrt:try:6} during its attempt $A$ in which $p$ does not receive an abort signal and 
	$p'$ has not even initiated its attempt $A'$.
	Since $\pcp = \refln{abrt:try:6}$, by Condition \ref{inv:abrt:cond6}, $\minarray[p] = (p, \tokp)$ in $B$.
	Let $B'$ be the earliest configuration following $B$ when $p'$ has performed contiguous normal steps upto Line~\refln{abrt:try:6} during its attempt $A'$.
	Therefore, by the same argument, $\minarray[p'] = (p', \tok{p'})$ in $B'$.
	It follows from the premise of the lemma that $\tokp < \tok{p'}$ (since $p$ performed contiguous normal steps upto Line~\refln{abrt:try:6} even before $A'$ started, applying Condition~\ref{inv:abrt:cond8} right at the configuration when $p$ completes Line~\refln{abrt:try:2}, $\tokp < \token \leq \tok{p'}$).
	Assume the lemma is false.
	Therefore, there is a configuration $C$ following $B'$ such that $p'$ entered the CS during attempt $A'$
	before $p$ entered the CS during attempt $A$.
	Therefore, $\pc{p'} = \refln{abrt:cs:1}$ and $\csowner = (1, p')$ in $C$.
	It follows that there is a configuration between $B'$ and $C$, call it $C'$, 
	such that, $\csowner \neq (1, p')$ in configurations $B'$ to the one just before $C'$,
	and $\csowner = (1, p')$ in configurations $C'$ to $C$.
	Let process $p''$ be the one that changed $\csowner$ to $(1, p')$ in $C'$.
	$p''$ could have changed $\csowner$ this way only at Line~\refln{abrt:prom:3}, since no other step sets the first bit of $\csowner$ to $1$.
	It follows that $p''$ read the record of $p'$ from the $\registry.\findmin()$ it executed at Line~\refln{abrt:prom:2}, 
	or equally, $p' = p''$ and $p'$ itself set $\peer{p'} = p'$ because it found the $\registry$ to be empty (this could happen because $p'$ either crashed in $\tryproc{p'}()$ or aborted).
	In either case, since $\tokp < \tok{p'}$, the $\findmin()$ by $p''$  at Line~\refln{abrt:prom:2} could have received $(p', \tok{p'})$ (or $(q, \infty)$, for some $q$, denoting $\registry$ to be empty)
	if and only if $\minarray[p] \neq (p, \tokp)$ (i.e., it is either $(p, \infty)$ or $(p, x)$ for a token $x$ higher than $\tok{p'}$).
	This implies that $p$ either already left the CS from the attempt $A$ or crashed some time after configuration $B$ (and thereby removing its own entry from $\registry$ at Line~\refln{abrt:abort:1})
	before $p''$ executed the $\findmin()$ at Line~\refln{abrt:prom:2}, and hence before the configuration $C$ is reached.
	Therefore, we have the lemma.
\end{proof}

\begin{lemma}[{\bf Starvation Freedom}] \label{lem:abrt:starvfreedom}
	In every fair infinite run in which every attempt contains 
	only finitely many crash steps,
	if a process $p$ is in the Try section in a configuration, $p$ is in a different section in a later configuration.
\end{lemma}
\begin{proof}
	Suppose the claim is false.
	Therefore, there is a fair infinite run in which a process $p$ starts an attempt and never leaves the $\tryprocp()$ procedure,
	i.e., it forever loops in the procedure at Line~\refln{abrt:try:6} after a certain configuration (this follows from the fact that every attempt contains only a finitely many crash steps).
	Let $C$ be the earliest configuration of the run such that $p$ forever waits at Line~\refln{abrt:try:6} after $C$,
	all other processes are either waiting with $p$ at Line~\refln{abrt:try:6} or are in the Remainder Section with $\status_p = \good$, and no process in the Recover, CS, or Exit Section.
	Such a configuration would exist because there are a finite number of processes each crashing finitely many times,
	and by Lemma~\ref{lem:abrt:fcfs} the algorithm satisfies the First Come First Served property.
	Therefore, without loss of generality, $p$ be the process so that no other process can enter the CS before $p$ enters it.
	It follows that $\csowner = (0, k)$, for some integer $k$, and $\registry[p] = (p, \tokp)$ from $C$ onwards.
	By Condition~\ref{inv:abrt:cond7} it follows that some process is either in CS or can be counted on to launch a waiting process into CS.
	This is a contradiction to our assumption that in $C$ all the processes that are active in an attempt are waiting at Line~\refln{abrt:try:6}.
\end{proof}

\begin{lemma}[{\bf Bounded Recovery to CS}] \label{lem:abrt:wfcsr}
	There is an integer $b$ such that if in any run any process $p$ 
	executes $\recoverprocp()$ without crashing and with $\status_p = \reccs$, 
	the method completes in at most $b$ steps of $p$, returning $\gotocs$.
\end{lemma}
\begin{proof}
	For $p$ to execute $\recoverprocp()$ with $\status_p = \reccs$, it must have crashed in the CS before.
	Let $C$ be a configuration prior to a crash step when $p$ is in the CS, i.e., $\pcp = \refln{abrt:cs:1}$ in $C$.
	By Condition~\ref{inv:abrt:cond7}, $\csowner = (1, p)$, and, by Condition~\ref{inv:abrt:cond5}, $\go[p] \neq -1$  in $C$.
	Without loss of generality, let $C'$ be the first configuration of a passage following $C$, such that, $p$ executes Line~\refln{abrt:abort:3} in this passage due to a call to $\abortprocp()$ from $\recoverprocp()$ in this passage.
	That is all passages, if any, between $C$ and $C'$ ended with a crash in $\recoverprocp()$ (or a crash within the nested call to $\abortprocp()$) before reaching and executing Line~\refln{abrt:abort:3}.
	Also note, by the description of $\status_p$, it will retain the value $\reccs$ even up to $C'$.
	Since no other process except for $p$ itself sets the value of $\go[p]$ to $-1$ at Lines~\refln{abrt:exit:6} and \refln{abrt:abort:4}, such a configuration is reachable in a bounded number of steps.
	It follows by the similar argument that $\csowner$ retains the value $(1, p)$ up to $C'$, because no other process can write a value $(0, k)$ at Line~\refln{abrt:exit:4}, for some integer $k$, 
	so that subsequently some process can perform the $\cas$ at Line~\refln{abrt:prom:3}.
	Thus starting at configuration $C'$, $p$ starts executing $\recoverprocp()$ and reaches Line~\refln{abrt:abort:3}.
	At Line~\refln{abrt:abort:3} $p$ notices that $\csowner = (1, p)$ and it returns from $\abortprocp()$ and subsequently from $\recoverprocp()$ with the value $\gotocs$.
	From an inspection of the algorithm we note that this happens within a constant number of steps from $C'$.
	The claim thus follows.
\end{proof}

\begin{lemma}[{\bf Critical Section Reentry}] \label{lem:abrt:csr}
	In any run, if a process $p$ crashes while in the CS, 
	no other process enters the CS until $p$ subsequently reenters the CS.
\end{lemma}
\begin{proof}
	Immediate from Lemma~\ref{lem:abrt:mutex} and \ref{lem:abrt:wfcsr}.
\end{proof}

\begin{lemma}[{\bf Bounded Recovery to Exit}] \label{lem:abrt:boundedrec2csexit}
	There is an integer $b$ such that if in any run any process $p$ 
	executes $\recoverprocp()$ without crashing and with $\status_p = \recexit$, 
	the method completes in at most $b$ steps of $p$.
\end{lemma}
\begin{proof}
	By an inspection of the algorithm, specifically that of $\recoverprocp()$, $\abortprocp()$, and $\promotep()$, we note that any execution path that $p$ takes after crashing with $\status_p = \recexit$,
	if $p$ executes $\recoverprocp()$ without crashing, then it completes the method in a constant number of steps.
\end{proof}

\begin{lemma}[{\bf Fast Recovery to Remainder}] \label{lem:abrt:fastrec2rem}
	There is a constant $b$ (independent of $|\procset|$) such that if in any run any process $p$ 
	executes $\recoverprocp()$ without crashing and with $\status_p \in \{\good, \recrem\}$, 
	the method completes in at most $b$ steps of $p$.
\end{lemma}
\begin{proof}
	By Condition~\ref{inv:abrt:cond5} we note that $\go[p] = -1$ when $\pcp = \refln{abrt:rem:1}$ and $\status_p \in \{ \good, \recrem \}$.
	It follows that if $p$ executes $\recoverprocp()$, it notices $\go[p] = -1$ at Line~\refln{abrt:rec:1} and immediately returns to the Remainder.
\end{proof}

\begin{lemma}[{\bf Bounded Recovery to Remainder}] \label{lem:abrt:boundedrec2rem}
	There is an integer $b$ such that if in any run
	$\recoverprocp()$, executed by a process $p$ with $\status_p = \rectry$, returns $\gotorem$,
	$p$ must have completed that execution of $\recoverprocp()$ in at most $b$ of its steps.
\end{lemma}
\begin{proof}
	From an inspection of the algorithm we note that any execution path that $p$ takes when it returns $\gotorem$ from $\recoverprocp()$,
	it must have done so in a constant number of steps from the latest step when it invoked $\recoverprocp()$.
	Thus the claim follows.
\end{proof}

\begin{lemma}[{\bf Bounded Abort}] \label{lem:abrt:boundedabort}
	There is an integer $b$ such that, for each $R, C, p$,
	if $\beta(R, p, C)$ is true,
	$\abortsigp$ stays $\true$ for ever (i.e., stays true in the suffix $R'$ of the run from $C$), 
	and $p$ executes steps without crashing (i.e., $p$ has no crash steps in $R'$),
	then $p$ enters either the CS or the remainder in at most $b$ of its steps (in $R'$).
\end{lemma}
\begin{proof}
	For this we note that the only \waittill loop that the algorithm has is at Line~\refln{abrt:try:6}.
	Since $\abortsigp$ stays $\true$ for ever after $C$,
	$p$ either notices that or sees that $\go[p] = 0$ at Line~\refln{abrt:try:6}.
	At Line~\refln{abrt:try:7}, if $p$ sees that $\go[p] = 0$, it moves to the CS, satisfying the condition.
	Otherwise, it invokes $\abortprocp()$ at Line~\refln{abrt:try:8}.
	From an inspection of $\abortprocp()$, we note that the procedure returns within a constant number of steps (i.e., $O(\log n)$ steps, where $n = |\procset|$)
	with a value of either $\gotocs$ or $\gotorem$.
	It follows that the claim holds.
\end{proof}

\begin{lemma}[{\bf No Trivial Aborts}] \label{lem:abrt:notrivialabort}
	In any run, if $\abortsigp$ is $\false$ when a process $p$ invokes $\tryprocp()$,
	$\abortsigp$ remains $\false$ forever, and $p$ executes steps without crashing,
	then $\tryprocp()$ does not return $\gotorem$.
\end{lemma}
\begin{proof}
	Since $p$ executes steps without crashing and $\abortsigp$ remains $\false$ forever in the run, 
	the only place $\tryprocp()$ could return $\gotorem$ is due to the nested call to $\abortprocp()$ at Line~\refln{abrt:try:8}.
	However, by Condition~\ref{inv:abrt:cond9}, we know that if $p$ gets past the \waittill loop at Line~\refln{abrt:try:6}, 
	then $\go[p] = 0$ when $\pcp = \refln{abrt:try:7}$ (since abort was not requested when $\tryprocp()$ was invoked and $\abortsigp$ remains $\false$ forever).
	It follows that $p$ returns $\gotocs$ in such a run.
\end{proof}

\subsection{RMR Complexity} \label{sec:rmrcomp}

We discuss the RMR complexity a process incurs per passage as follows.
As described in Lemma 2 of Jayanti and Joshi's work \cite{jayanti:fcfsmutex}, the $\registry.\mwrite()$ operation incurs $O(\min(k, \log n))$ RMRs on both CC and DSM machines, where $k$ is the maximum point contention
during the $\registry.\mwrite()$ operation.
On DSM machines, when the variable $\go[p]$ is hosted in $p$'s memory partition,
any step of the algorithm other than $\registry.\mwrite()$ (at Lines~\refln{abrt:try:4}, \refln{abrt:exit:1}, \refln{abrt:abort:1}) incurs a constant RMR.
Therefore, on DSM machines our algorithm incurs $O(\min(k,\log n))$ RMR per passage.
On CC machines, similarly, it would be tempting to believe that all these other operations incur constant RMRs,
however, it is not so due to the following.
On Strict-CC machines where a failed \cas\ could incur an RMR, the RMR complexity shoots up to $O(n)$ for the following reason.
There could be $n/2$ processes that are waiting to execute Line~\refln{abrt:prom:6} to perform a \cas\ on $\go[p]$.
Out of these processes only one succeeds and the rest fail.
However, each failed \cas\ still incurs an RMR.
Therefore, on Strict-CC machines our algorithm incurs $O(n)$ RMR per passage.
To summarize, the algorithm incurs $O(\min(k, \log n))$ RMRs per passage on DSM and Relaxed-CC machines and $O(n)$ RMRs per passage on Strict-CC machines.

For an attempt having $f$ failures, the implementation of $\registry$ taken from Jayanti and Joshi's work \cite{jayanti:fcfsmutex} would incur $O(f + \min(k,\log n))$ RMRs for the $\registry.\mwrite()$ operation.
%There is also a straightforward implementation directly based on $f$-arrays \cite{farrays} and satisfying %correctness conditions for recoverable objects \cite{attiya:rlin},
%which can use the checkpointing idea from the $\registry$ implementation in \cite{jayanti:fcfsmutex} to get a $O(f %+ \log n)$ RMR implementation for $\registry.\mwrite()$.
Therefore, the algorithm incurs $O(f + \log n)$ RMRs per attempt on DSM and Relaxed-CC machines and $O(f + n)$ RMRs per attempt on Strict-CC machines in the presence of $f$ crashes in an attempt.

\subsection{Proof of Invariant}

\begin{lemma}
	The algorithm in Figure~\ref{algo:abrt} satisfies the invariant (i.e., the conjunction of all the conditions)
	stated in Figure~\ref{inv:abrt}, i.e., the invariant holds in every configuration of every run of the algorithm.
\end{lemma}
\begin{proof} 
	 The proof is by induction, but it is omitted because of the page limitation on the submission. The full version of this paper, including this proof, can be found at 
\linebreak
http://people.csail.mit.edu/siddhartha/archive.html
\end{proof}

\subsection{Main theorem}
The theorem below summarizes the result of our paper.

\begin{theorem} \label{thm:abrt}
	The algorithm in Figure~\ref{algo:abrt} is an abortable recoverable mutual exclusion algorithm 
	for $n$ processes and satisfies properties P1-P12 stated in Section~\ref{ch2:spec}. 
	A process incurs $O(\min(k,\log n))$ RMRs per passage on DSM and Relaxed-CC machines and $O(n)$ RMRs per passage on Strict-CC machines.
	In presence of $f$ crashes during an attempt, a process incurs $O(f + \min(k,\log n))$ RMRs per attempt on DSM and Relaxed-CC machines and $O(f + n)$ RMRs per attempt on Strict-CC machines.
\end{theorem}

\vspace{0.2in}
\noindent
{\bf Acknowledgment}: We thank Siddhartha Jayanti for his careful reading and critical comments on the first three sections of this submission and the Netys '19 reviewers for their feedback.

%% Bibliography
\bibliographystyle{acm}
\bibliography{recoverable}

\end{document}
% end of file template.tex